\documentclass[aps,amsmath,amssymb,prl,twocolumn,superscriptaddress,floatfix]{revtex4}
\usepackage[dvipdfmx]{graphicx}
\usepackage{amsmath, amssymb, amsthm}
\usepackage{mathtools}
\usepackage{color}
\usepackage{listings}
\usepackage{braket}
\usepackage{hyperref}

\DeclareMathOperator{\tr}{Tr}
\DeclareMathOperator{\id}{id}
\allowdisplaybreaks[4]

\theoremstyle{definition}
\newtheorem{theorem}{Theorem}
\newtheorem{lemma}[theorem]{Lemma}

\newtheorem{proposition}[theorem]{Proposition}
\theoremstyle{definition}

\theoremstyle{remark}

\begin{document}

\title{Spread quantum information in one-shot quantum state merging}

\author{Hayata Yamasaki}
\email{yamasaki@eve.phys.s.u-tokyo.ac.jp}
\affiliation{Department of Physics, Graduate School of Science, The University of Tokyo, 7--3--1 Hongo, Bunkyo-ku, Tokyo 113--0033, Japan}

\author{Mio Murao}
\email{murao@phys.s.u-tokyo.ac.jp}
\affiliation{Department of Physics, Graduate School of Science, The University of Tokyo, 7--3--1 Hongo, Bunkyo-ku, Tokyo 113--0033, Japan}

\date{\today}

\begin{abstract}
  We prove the difference between the minimal entanglement costs in quantum state merging under one-way and two-way communication in a one-shot scenario, whereas they have been known to coincide asymptotically. While the minimal entanglement cost in state merging under one-way communication is conventionally interpreted to characterize partial quantum information conditioned by quantum side information, we introduce a notion of \textit{spread quantum information} evaluated by the corresponding cost under two-way communication. Spread quantum information quantitatively characterizes how nonlocally one-shot quantum information is spread, and it cannot be interpreted as partial quantum information.
\end{abstract}

\maketitle

\textit{Introduction.}---Analyses of quantum communication tasks quantitatively characterize the nature of information encoded in quantum systems.
Conventionally, tasks in \textit{asymptotic scenarios} considering an independent and identically distributed (IID) source have been used for such characterizations; \textit{e.g.}, Schumacher compression~\cite{S10} quantifies information per quantum state from an IID source in terms of the minimum amount of quantum communication for transferring these states, which is evaluated by the quantum entropy.
By contrast, recently developing one-shot quantum information theory~\cite{T11} aims to analyze quantum communication tasks for a single copy of state, and such \textit{one-shot scenarios} are more relevant to practical situations on small and intermediate scales~\cite{P4} of up to several dozens of qubits.
Consider two parties, namely a sender $A$ and a receiver $B$, aiming to transfer $A$'s one-qubit state $\Ket{\psi_{\boldsymbol\alpha}}\coloneqq\alpha_0\Ket{0}^A+\alpha_1\Ket{1}^A$ from $A$ to $B$, where $\left\{\Ket{0},\Ket{1}\right\}$ is a fixed orthogonal basis, and $\alpha_0,\alpha_1\in\mathbb{C}$ are arbitrary and unspecified complex coefficients satisfying $\sum_l{\left|\alpha_l\right|}^2=1$.
In this one-shot scenario, the required amount of quantum communication is one qubit, which can be interpreted to quantify \textit{one-shot quantum information} of $\Ket{\psi_{\boldsymbol\alpha}}$ similarly to the asymptotic scenario.
In this sense, one-shot quantum information of $A$'s arbitrary multi-qubit state in the form of $\alpha_0\Ket{\psi_0}^A+\alpha_1\Ket{\psi_1}^A$ is also one qubit, where $\left\{\Ket{\psi_0},\Ket{\psi_1}\right\}$ is a known set of fixed multi-qubit mutually orthogonal states corresponding to above $\left\{\Ket{0},\Ket{1}\right\}$, and quantum information is represented as a superposition of these multi-qubit states.

When quantum information is encoded in a composite system, such as above $\alpha_0\Ket{\psi_0}+\alpha_1\Ket{\psi_1}$,
subsystems can be distributed between the two distant parties $A$ and $B$.
We ask how to characterize nonlocal properties of one-shot quantum information \textit{spread nonlocally} between $A$ and $B$, \textit{i.e.}, represented as an arbitrary superposition of $D$ shared mutually orthogonal states $\sum_{l=0}^{D-1}\alpha_l\Ket{\psi_l}^{AB}$.
A task of transferring $A$'s part of $\sum_{l=0}^{D-1}\alpha_l\Ket{\psi_l}^{AB}$ to $B$ without destroying coherence is equivalent to quantum state merging~\cite{Y12} (simply called state merging in the following),
a communication task where $A$ and $B$ initially share a purified state $\Ket{\psi}^{RAB}\coloneqq\frac{1}{\sqrt{D}}\sum_{l=0}^{D-1}\Ket{l}^R\otimes\Ket{\psi_l}^{AB}$ in terms of a reference $R$ with its computational basis ${\left\{\Ket{l}^R\right\}}_{l=0,\ldots,D-1}$, and a protocol transfers $A$'s part of $\Ket{\psi}^{RAB}$ to $B$, keeping coherence between $R$ and $B$~\cite{H3,H4}.
The final state in state merging can be written as $\Ket{\psi}^{RB^\prime B}$, where $B^\prime$ is $B$'s system corresponding to $A$.
These situations of state merging ubiquitously appear in distributed quantum information processing~\cite{W8,W9,W10,Y13,W12}, multipartite entanglement transformations~\cite{A3,D8,D9,H8,A8,Y8,S4},
and analyses of a family of other communication tasks in quantum Shannon theory~\cite{W5,A2,O1,D1,H5,D4,A7,H11,P3,L6,L7}.
State merging is also used for attempts to understand physical phenomena in quantum thermodynamics~\cite{F1} and quantum gravity~\cite{A1}.

Originally in Refs.~\cite{H3,H4}, state merging is formulated in the asymptotic scenario, where $n$ copies of $\Ket{\psi}^{RAB}$ are given and $A$'s part is transferred to $B$ within a vanishing error in fidelity as $n\to\infty$.
In this asymptotic scenario, state merging has been interpreted to characterize \textit{partial quantum information}~\cite{H3} in analogy to the classical Slepian-Wolf problem~\cite{S7}, where $A$ and $B$ are given classical messages $X$ and $Y$ respectively, and use classical communication from $A$ to $B$ to let $B$ know $X$. If $Y$ is correlated with $X$, partial knowledge on $X$ is provided to $B$ by $Y$, which is called $B$'s \textit{side information}, and the conditional entropy $H\left(X|Y\right)$ characterizes how much \textit{partial information} on $X$ conditioned by $Y$ should be additionally transferred from $A$ to $B$.
In state merging of $\Ket{\psi}^{RAB}$,
quantum communication is achieved by means of quantum teleportation~\cite{B5}, using local operations and classical communication (LOCC)~\cite{C7} assisted by entanglement resources shared between $A$ and $B$.
The amount of entanglement required for state merging is called \textit{entanglement cost},
which is the cost to be minimized.

The minimal entanglement cost in state merging of $n$ copies of $\Ket{\psi}^{RAB}$ asymptotically converges to the conditional quantum entropy ${H\left(A|B\right)}_{\psi}$ per copy~\cite{H3,H4}.
Even if $A$ and $B$ are allowed to perform \textit{two-way} LOCC using classical communication both from $A$ to $B$ and from $B$ to $A$,
this minimal cost can be achieved by only \textit{one-way} LOCC using one-way classical communication from $A$ to $B$; that is, $B$'s preprocessing and classical communication from $B$ to $A$ do not contribute to reducing this cost.
Analogously to the classical Slepian-Wolf problem, $B$'s part of $\Ket{\psi}^{RAB}$ is called \textit{quantum side information} at $B$.
Similar notions of quantum side information are widely used in other contexts, such as entropic uncertainty relations~\cite{C12}, state exchange~\cite{L7,O1}, and classical-quantum Slepian-Wolf problems~\cite{D10,T5,R3,T9,L1,M5,C11}.
Entanglement cost of an optimal one-way LOCC protocol for state merging has been interpreted to characterize how much \textit{partial quantum information} conditioned by $B$'s quantum side information should be additionally transferred from $A$ to $B$.

In this letter, we show that the minimal entanglement cost in state merging under one-way LOCC and that under two-way LOCC can be \textit{different} in a one-shot scenario, whereas they coincide in the asymptotic scenario~\cite{H3,H4}.
State merging and its generalized task, state redistribution~\cite{D2,D3}, have also been defined and analyzed in various one-shot scenarios~\cite{Y12,Y9,B9,B12,D7,D6,H10,B10,D5,M,N3,A4,A5,B15,B13,A16,A17}, as well as other derivatives~\cite{B1,B2,S11,S12,S14,A10,A11}.
In the one-shot scenarios, state merging can be achieved using techniques of one-shot decoupling~\cite{D6}, the convex-split lemma~\cite{A4}, or the Koashi-Imoto decomposition~\cite{Y12}.
However, these existing techniques employ \textit{only one-way} communication similarly to the asymptotic scenario.
Although two-way LOCC protocols may outperform one-way LOCC protocols in general,
\textit{no example} of separation between one-way LOCC and two-way LOCC in entanglement cost of one-shot state merging has been known~\cite{B10}, due to hardness of evaluating the minimal entanglement cost in one-shot scenarios.
In contrast, we prove an advantage of two-way LOCC over one-way LOCC in reducing entanglement cost in a one-shot scenario of state merging, by showing an instance.
As a result, the following two different (but asymptotically coinciding) notions of \textit{quantum information} characterized by one-shot state merging are suggested, as will be discussed after showing our results:
\begin{itemize}
  \item\textit{Partial quantum information conditioned by $B$'s quantum side information:} The minimal entanglement cost in state merging under \textit{one-way} LOCC, in analogy to that in the classical Slepian-Wolf problem as discussed in Ref.~\cite{H3};
  \item\textit{Spread quantum information for $B$:} The minimal entanglement cost in state merging under \textit{two-way} LOCC, characterizing how much shared entanglement is needed to concentrate on $B$ quantum information initially spread between $A$ and $B$.
\end{itemize}

\begin{figure}[t!]
  \centering
  \includegraphics[width=8.6cm]{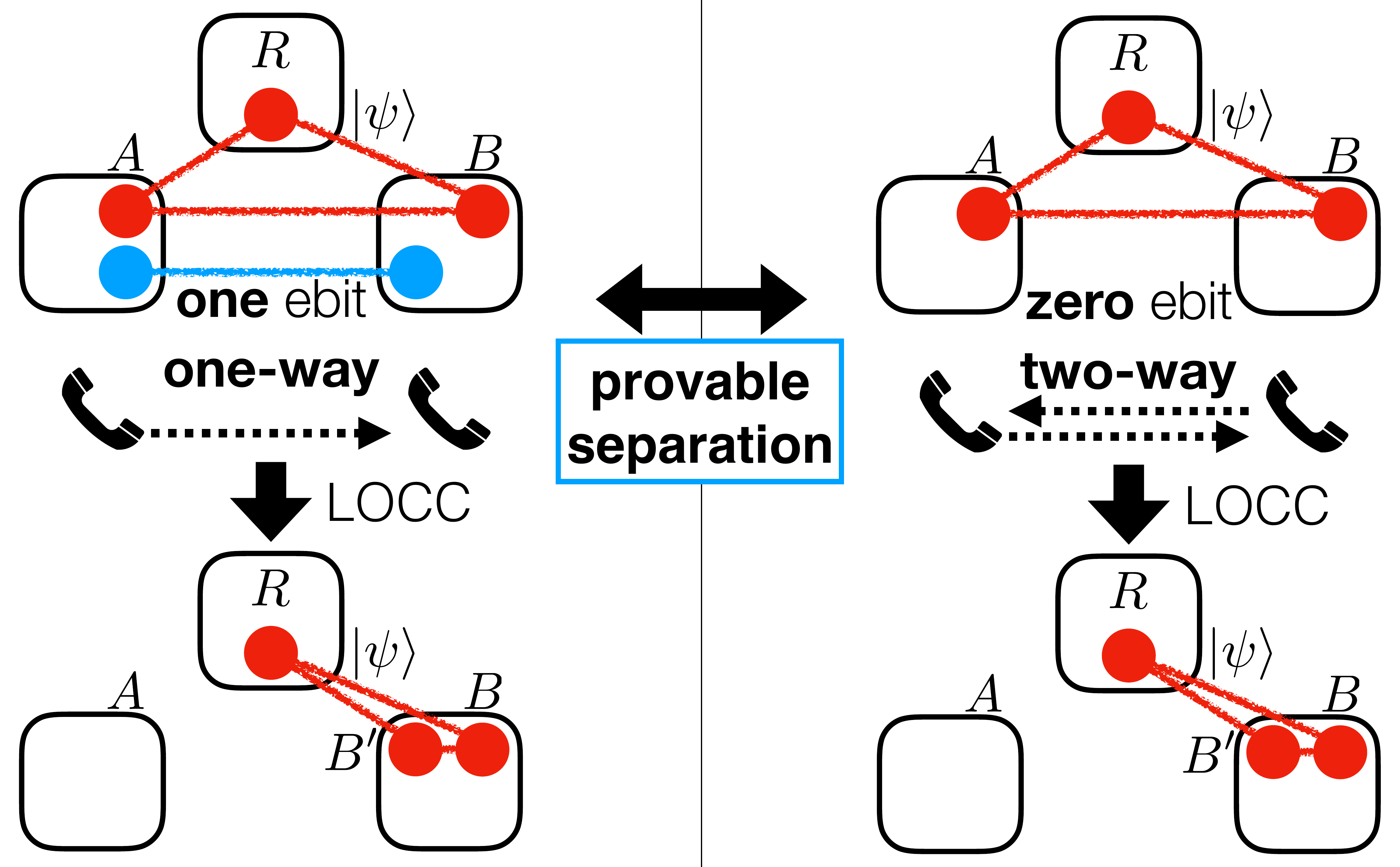}
  \caption{\label{fig:result}One-way LOCC and two-way LOCC protocols for a one-shot state merging of $\Ket{\psi}^{RAB}$ defined as Eq.~\eqref{eq:phi} represented by the red circles, where classical communication is represented by the dotted arrows. While the optimal one-way LOCC protocol for this task requires nonzero entanglement cost (one ebit, represented by the connected blue circles), there exists a two-way LOCC protocol achieving zero entanglement cost, leading a provable separation.}
\end{figure}

\begin{table*}[t!]
  \centering
  \caption{\label{table:compare}Whether or not there exists the separation between one-way LOCC and two-way LOCC in terms of achievability of several quantum tasks in asymptotic and one-shot scenarios.}
  \begin{ruledtabular}
  \begin{tabular}{@{}llllll@{}}
    task & \begin{tabular}{@{}l@{}} state transformation\\(bipartite pure) \end{tabular} & state splitting & \begin{tabular}{@{}l@{}}state merging\\ (non-catalytic) \end{tabular} & \begin{tabular}{@{}l@{}} entanglement\\ distillation \end{tabular} & \begin{tabular}{@{}l@{}}local state\\ discrimination \end{tabular} \\ \hline
    asymptotic scenario & No~\cite{B3}. & No~\cite{A2}. & No~\cite{H3,H4}. & Yes~\cite{B14}. & Yes~\cite{O3}. \\
    one-shot scenario & No~\cite{N2}. & No~\cite{Y12}. & Yes (Theorem~\ref{thm:result}). & Yes~\cite{C2}. & Yes~\cite{G4,C4,O2,C5,N4,T7,T8,C3}. \\
  \end{tabular}
\end{ruledtabular}
\end{table*}

\textit{Separation between one-way LOCC and two-way LOCC in a one-shot state merging.}---The task of one-shot quantum state merging analyzed in this letter is defined as follows.
Given a state $\Ket{\psi}^{RAB}$ and an error $\epsilon\geqq 0$, define a one-shot scenario of state merging of $\Ket{\psi}^{RAB}$ within $\epsilon$ as a task of $A$ and $B$ performing LOCC assisted by a maximally entangled state $\Ket{\Phi_K}^{AB}\coloneqq\frac{1}{\sqrt{K}}\sum_{l=0}^{K-1}\Ket{l}^A\otimes\Ket{l}^B$ with Schmidt rank $K$ to transform $\Ket{\psi}^{RAB}\otimes\Ket{\Phi_K}^{AB}$ into a final state ${\tilde\psi}^{RB^\prime B}$ satisfying the fidelity condition $F^2\left({\tilde\psi}^{RB^\prime B},\Ket{\psi}\Bra{\psi}^{RB^\prime B}\right)\coloneqq\Bra{\psi}{\tilde\psi}\Ket{\psi}\geqq 1-\epsilon^2$, where entanglement cost $\log_2 K$ is to be minimized.
We call this task \textit{non-catalytic approximate state merging}, or if obvious, state merging.
This task is a smoothed version of \textit{non-catalytic exact state merging} defined in Ref.~\cite{Y12},
while our results also apply to non-catalytic exact state merging. Such non-catalytic use of shared entanglement resources is relevant to one-shot scenarios implemented by small- and intermediate-scale quantum computers~\cite{P4} of up to several dozens of low-noise qubits where a large amount of entanglement catalyst cannot be stored faithfully.
Our main result (Fig.~\ref{fig:result}) is shown as follows.

\begin{theorem}
\label{thm:result}
  \textit{Separation of one-way and two-way LOCC in a one-shot state merging.}
  There exists a state $\Ket{\psi}^{RAB}$ (defined later in Eq.~\eqref{eq:phi}) and a nonzero error threshold $\epsilon_0>0$ such that for any $\epsilon\in\left[0,\epsilon_0\right]$, the following hold.
  \begin{enumerate}
  \item \textit{No one-way LOCC protocol} for non-catalytic approximate state merging of $\Ket{\psi}^{RAB}$ within $\epsilon$ achieves zero entanglement cost, while \textit{one} ebit of entanglement cost, \textit{i.e.}, $\log_2 K = 1$, is achievable.
  \item There exists \textit{a two-way LOCC protocol} for non-catalytic approximate state merging of $\Ket{\psi}^{RAB}$ within $\epsilon$ achieving \textit{zero} entanglement cost, \textit{i.e.}, $\log_2 K = 0 < 1$.
\end{enumerate}
\end{theorem}

Regarding provable separations between one-way LOCC and two-way LOCC in achievability of quantum tasks, only a few examples are known to date, such as entanglement distillation and local state discrimination, as summarized in Table~\ref{table:compare}.
While the set of one-way LOCC maps is strictly included in that of two-way LOCC maps~\cite{C7}, this difference does not necessarily affect achievability of some tasks; \textit{e.g.}, one-way LOCC suffices for deterministic transformations between two fixed bipartite pure states and state splitting, the inverse communication task of state merging.
Among the known separations, the separation in local state discrimination based on hypothesis testing is first proven in a one-shot scenario~\cite{O2}, and later in Ref.~\cite{O3}, it is shown that the separation \textit{does survive} in the corresponding asymptotic scenario.
In contrast to such known separations shown in both asymptotic and one-shot scenarios, Theorem~\ref{thm:result} on state merging provides a case where a provable separation in a one-shot scenario \textit{does not survive} in the asymptotic scenario.
Note that in the asymptotic scenario of state merging of $\Ket{\psi}^{RAB}$, even if entanglement catalyst is allowed in the definition itself, the amount of entanglement catalyst necessary in the asymptotically optimal protocols~\cite{H4,B20} can be \textit{arbitrarily close to zero} per copy of $\Ket{\psi}^{RAB}$.
It is unknown how common such examples exhibiting the separation between between one-way LOCC and two-way LOCC in one-shot state merging are~\footnote{Our construction of $\Ket{\psi}$ in Theorem~\ref{thm:result} uses three $\left(11\times 11\right)$-dimensional mutually orthogonal pure states $\Ket{\psi_0}$, $\Ket{\psi_1}$, and $\Ket{\psi_2}$ as shown in Eqs.~\eqref{eq:s} and~\eqref{eq:phi}, while we do not optimize these dimensions. As for the number of these states, \textit{i.e.}, three, as long as our proof technique is used for separation between one-way LOCC and two-way LOCC, two states are not sufficient since local state discrimination for any two states is achievable by one-way LOCC~[J.\ Walgate, A.\ J.\ Short, L.\ Hardy, and V.\ Vedral, Phys.\ Rev.\ Lett.\ \textbf{85}, 4972 (2000).].}.

\textit{Connection between state merging and local state discrimination.}---For proving Theorem~\ref{thm:result} on the separation between one-way LOCC and two-way LOCC in a one-shot state merging, we need a no-go theorem that is \textit{only applicable to one-way} LOCC and is \textit{provably false for two-way} LOCC\@.
The existing proof techniques used in Refs.~\cite{H4,B10,Y12} for obtaining lower bounds of entanglement cost in state merging are not sufficient, since these techniques are based on no-go theorems applicable to any LOCC maps.
Hence, another proof technique than these existing ones has to be established.

In our proof of Theorem~\ref{thm:result}, local state discrimination plays an essential role.
In local state discrimination, two parties $A$ and $B$ initially share an unknown state $\Ket{\psi_l}^{AB}$ given from a known set ${\left\{\Ket{\psi_l}^{AB}\right\}}_{l=0,\ldots,D-1}$ of $D$ mutually orthogonal pure states, and the task aims to determine the index $l$ of $\Ket{\psi_l}^{AB}$ with unit probability by an LOCC measurement.
There exists a set of such states for which local state discrimination is not achievable by one-way LOCC but is achievable by two-way LOCC, which is called a $2$-LOCC set.
References~\cite{N4,T7,T8} provide $2$-LOCC sets for any possible dimensional systems.

State merging and local state discrimination are related in the sense that achievability of state merging implies that of local state discrimination.
If there exists a protocol exactly achieving state merging of a state having the Schmidt decomposition $\Ket{\psi}^{RAB}\coloneqq\frac{1}{\sqrt{D}}\sum_{l=0}^{D-1}\Ket{l}^R\otimes\Ket{\psi_l}^{AB}$ at zero entanglement cost,
then this protocol transforms any superposition of the $D$ mutually orthogonal states ${\left\{\Ket{\psi_l}^{AB}\right\}}_{l}$ into that of ${\left\{\Ket{\psi_l}^{B^\prime B}\right\}}_{l}$~\cite{Y12}, \textit{i.e.},
\begin{equation}
  \label{eq:relative}
  \sum_{l=0}^{D-1}\alpha_l\Ket{\psi_l}^{AB}\xrightarrow{\textup{LOCC}}\sum_{l=0}^{D-1}\alpha_l\Ket{\psi_l}^{B^\prime B}.
\end{equation}
Local state discrimination for ${\left\{\Ket{\psi_l}^{AB}\right\}}_{l}$ can be achieved by first performing the protocol for state merging of $\Ket{\psi}^{RAB}$ to transform $\Ket{\psi_l}^{AB}$ into $\Ket{\psi_l}^{B^\prime B}$ for any $l$, and then performing $B$'s measurement for discriminating $B$'s mutually orthogonal states ${\left\{\Ket{\psi_l}^{B^\prime B}\right\}}_{l}$.
Note that a similar connection also holds in the asymptotic scenario~\cite{A9}.

In contrast, achievability of local state discrimination does not necessarily imply that of state merging if a protocol achieving local state discrimination uses a technique of \textit{elimination}, \textit{i.e.}, the measurement for excluding some of the possibilities in ${\left\{\Ket{\psi_l}^{AB}\right\}}_{l}$.
A protocol for local state discrimination using elimination cannot be utilized for state merging because elimination destroys coherence of the superposition of ${\left\{\Ket{\psi_l}\right\}}_{l}$.
As for the known $2$-LOCC sets,
two-way LOCC protocols for local state discrimination in Refs.~\cite{N4,T7,T8} require elimination, and hence, \textit{are not applicable} to state merging straightforwardly.

We identify a $2$-LOCC set for which a two-way LOCC protocol for local state discrimination can be constructed \textit{without elimination},
and we construct the corresponding two-way protocol for state merging.
Consider a set ${\left\{\Ket{\psi_l}^{AB}\in\mathbb{C}^{11}\otimes\mathbb{C}^{11}\right\}}_{l=0,1,2}$ of three mutually orthogonal states
\begin{equation}
  \label{eq:s}
  \begin{split}
    \Ket{\psi_0}^{AB}\coloneqq&\sqrt{\frac{2}{11}}\Ket{\Phi_2}^{AB}
    \oplus\sqrt{\frac{9}{11}}\Ket{\Phi_9}^{AB},\\
    \Ket{\psi_1}^{AB}\coloneqq&\sqrt{\frac{2}{11}}\gamma_1 X_2^A\Ket{\Phi_2}^{AB}
    \oplus\sqrt{\frac{9}{11}}{\left(X_9^A\right)}^3\Ket{\Phi_9}^{AB},\\
    \Ket{\psi_2}^{AB}\coloneqq&\sqrt{\frac{2}{11}}\gamma_2 Z_2^A\Ket{\Phi_2}^{AB}
    \oplus\sqrt{\frac{9}{11}}{\left(X_9^A\right)}^6\Ket{\Phi_9}^{AB},
  \end{split}
\end{equation}
where each subsystem is decomposed into subspaces $\mathbb{C}^{11}=\mathbb{C}^{2}\oplus\mathbb{C}^{9}$,
$X_k^A$ and $Z_k^A$ are the generalized Pauli operator~\cite{W5} on a subspace $\mathbb{C}^k$ of $A$'s system for $A$'s part of $\Ket{\Phi_k}^{AB}\coloneqq\frac{1}{\sqrt{k}}\sum_{l=0}^{k-1}\Ket{l}^A\otimes\Ket{l}^B$, and $\gamma_1$ and $\gamma_2$ are nonreal complex numbers satisfying ${\left|\gamma_1\right|}^2=1$, ${\left|\gamma_2\right|}^2=1$, and $\gamma_2\neq\pm\textup{i}\gamma_1^2$.
Define a tripartite state
\begin{equation}
\label{eq:phi}
  \Ket{\psi}\coloneqq\frac{1}{\sqrt{3}}\sum_{l=0}^{2}\Ket{l}^R\otimes\Ket{\psi_l}^{AB},
\end{equation}
where ${\left\{\Ket{\psi_l}^{AB}\right\}}_{l=0,1,2}$ is given by Eq.~\eqref{eq:s}.
This state $\Ket{\psi}^{RAB}$ yields Theorem~\ref{thm:result} as follows.

\textit{Proof of the first statement in Theorem~\ref{thm:result}.}---The set ${\left\{\Ket{\psi_l}^{AB}\right\}}_{l=0,1,2}$ defined as Eq.~\eqref{eq:s} is shown to be a $2$-LOCC set~\cite{N4},
and hence, impossibility of local state discrimination by one-way LOCC yields impossibility of non-catalytic \textit{exact} state merging of $\Ket{\psi}^{RAB}$ defined as Eq.~\eqref{eq:phi} at zero entanglement cost by one-way LOCC\@.
Since the set of one-way LOCC maps is compact~\cite{C7},
this impossibility of non-catalytic exact state merging by one-way LOCC implies that there exists a sufficiently small but nonzero error $\epsilon>0$ such that non-catalytic \textit{approximate} state merging of $\Ket{\psi}^{RAB}$ within $\epsilon$ is still impossible at zero entanglement cost by one-way LOCC\@.
Note that the no-go theorem on local state discrimination by one-way LOCC in Ref.~\cite{N4} does not straightforwardly generalize to scenarios where catalytic use of entanglement is allowed, due to the existence of entanglement discrimination catalysis~\cite{Y15}.
We also construct a one-way LOCC protocol for state merging of $\Ket{\psi}^{RAB}$ achieving one ebit of entanglement cost and zero error, \textit{i.e.}, $\log_2 K =1$ and $F^2\left(\tilde\psi,\Ket{\psi}\Bra{\psi}\right)=1$, based on a method established in Ref.~\cite{Y12} using the Koashi-Imoto decomposition~\cite{K3,H6,K5,W4} (See Supplemental Material for detail), which yields the first statement in Theorem~\ref{thm:result}.

\textit{Proof of the second statement in Theorem~\ref{thm:result}.}---We construct a two-way LOCC protocol for state merging of $\Ket{\psi}^{RAB}$ defined as Eq.~\eqref{eq:phi} achieving zero entanglement cost and zero error, \textit{i.e.}, $\log_2 K =0$ and $F^2\left(\tilde\psi,\Ket{\psi}\Bra{\psi}\right)=1$.
While three maximally entangled two-qubit states $\left\{\Ket{\Phi_2}^{AB},\gamma_1 X_2^A\Ket{\Phi_2}^{AB},\gamma_2 Z_2^A\Ket{\Phi_2}^{AB}\right\}$ used as part of Eq.~\eqref{eq:s} cannot be discriminated by any LOCC measurement~\cite{G10},
our two-way LOCC protocol begins with $B$'s measurement represented as measurement (Kraus) operators ${\left\{M_j^B\right\}}_{j=0,1,2}$ satisfying $\sum_j{M_j^B}^\dag M_j^B=\openone$, in order for the additional terms on $A$'s subspace $\mathbb{C}^9$ in Eq.~\eqref{eq:s} to be mutually orthogonal.
Using this orthogonality, $A$ can perform a thirty-three-outcome measurement represented as ${\left\{M_{k|j}^A\right\}}_{k=0,\ldots,32}$ conditioned by $B$'s measurement outcome $j$ and satisfying $\sum_k{M_{k|j}^A}^\dag M_{k|j}^A=\openone$ for each $j$,
so that for each measurement outcome $j$ and $k$, mutually orthogonal states ${\left\{\Ket{\psi_l}^{AB}\right\}}_{l=0,1,2}$ defined as Eq.~\eqref{eq:s} can be transformed into mutually orthogonal states of $B$.
Then, $B$'s local isometry correction conditioned by $j$ and $k$ provides $\Ket{\psi}^{RB^\prime B}$ (See Supplemental Material for detail), which yields the second statement in Theorem~\ref{thm:result}.

\textit{Partial quantum information and spread quantum information.}---A conventional interpretation of entanglement cost in state merging as \textit{partial quantum information}~\cite{H3} is based on the one-way communication picture in analogy to the classical Slepian-Wolf problem,
but Theorem~\ref{thm:result} implies that this interpretation is not straightforwardly applicable to the entanglement cost under two-way LOCC, which can be different in quantity from that under one-way LOCC\@.
Instead, we consider another aspect of state merging in association with local state discrimination; that is, state merging and local state discrimination can be viewed as \textit{distributed decoding} of information encoded in a shared quantum state.
In local state discrimination for ${\left\{\Ket{\psi_l}^{AB}\right\}}_l$, the index $l$ can be regarded as classical information encoded in $\Ket{\psi_l}^{AB}$, and local state discrimination aims to decode this classical information by LOCC\@.
In the same way, state merging of $\Ket{\psi}^{RAB}$ can be regarded as distributed decoding of \textit{quantum information}~\cite{Y13} by entanglement-assisted LOCC, in the sense that a protocol for state merging decodes an arbitrary superposition of fixed mutually orthogonal states shared between $A$ and $B$ into the same superposition of $B$'s states, as shown in Formula~\eqref{eq:relative}.
These notions of information may be \textit{nonlocally} encoded in the shared quantum state~\cite{B6}, in the sense that neither $A$ nor $B$ has local access to such nonlocally encoded information.

From this viewpoint, a nonlocal property of the map $\mathcal{D}^{AB\to B}\left(\sum_l\alpha_l\Ket{\psi_l}^{AB}\right)=\sum_l\alpha_l\Ket{\psi_l}^{B^\prime B}$ for decoding quantum information nonlocally spread in $\sum_l\alpha_l\Ket{\psi_l}^{AB}$ is characterized by the minimal entanglement cost in state merging, which we name \textit{spread quantum information} for $B$.
This map $\mathcal{D}^{AB\to B}$ is an isometry map that can be defined for every given $\Ket{\psi}^{RAB}$~\cite{Y12,Y13}.
Note that if catalytic use of entanglement is allowed, negative entanglement cost can also be viewed as a net gain of entanglement resources from the redundant part of $\sum_l\alpha_l\Ket{\psi_l}^{AB}$ as discussed in Ref.~\cite{Y12}, and the gained entanglement can be used as a resource for distributed decoding in future.
Spread quantum information for $B$ characterizes how nonlocally for $B$ the encoded one-shot quantum information is spread, while it is not invariant under reversal of $A$ and $B$ similarly to partial quantum information conditioned by $B$'s quantum side information.

\textit{Conclusion.}---
We showed that in a one-shot scenario of state merging, the minimal entanglement cost under \textit{one-way} LOCC and that under \textit{two-way} LOCC can be different, whereas these costs coincide in the asymptotic scenario.
Our analysis employs an interconnection between state merging and local state discrimination to demonstrate a provable separation between one-way LOCC and two-way LOCC in entanglement cost of state merging (a comparison to other quantum tasks is summarized in Table~\ref{table:compare}).
Based on this interconnection, one-shot state merging can be viewed as distributed decoding of one-shot quantum information nonlocally encoded and spread in a composite quantum system.
The minimal entanglement cost in one-shot state merging under two-way LOCC, named \textit{spread quantum information}, quantitatively characterizes how nonlocally one-shot quantum information is spread, and it cannot be interpreted as partial quantum information.
On the contrary, if we regard our two-way LOCC protocol for state merging as $B$'s preprocessing of quantum side information and backward classical communication from $B$ to $A$ followed by one-way LOCC from $A$ to $B$, our results suggest that such preprocessing can be indispensable for minimizing entanglement cost in a one-shot state merging from $A$ to $B$.
In this regard, our results open the way to further research on how to utilize quantum side information and two-way multi-round communication in quantum information theory.

\acknowledgments{%
  This work was supported by Grant-in-Aid for JSPS Research Fellow, JSPS KAKENHI Grant Number 26330006, 15H01677, 16H01050, 17H01694, 18H04286, 18J10192, and the Q-LEAP project of the MEXT Japan.
}

\bibliography{citation_bibtex}

\clearpage
\appendix
\widetext%
\begin{center}
\textbf{\large Supplemental Materials}
\end{center}

\section{Notations}

We represent a system indexed by $A$ as a finite-dimensional Hilbert space $\mathcal{H}^A$.
Superscripts of an operator or a vector represent the indices of the corresponding Hilbert spaces, e.g., $\psi^{RA}$ on $\mathcal{H}^R\otimes\mathcal{H}^A$ for a mixed state and $\Ket{\psi}^{RAB}\in\mathcal{H}^R\otimes\mathcal{H}^A\otimes\mathcal{H}^B$ for a pure state.
We may write an operator corresponding to a pure state as $\psi^{RAB}\coloneqq\Ket{\psi}\Bra{\psi}^{RAB}$.
A reduced state may be represented by superscripts, such as $\psi^{B}\coloneqq\tr_{RA}\psi^{RAB}$.
The identity operator and the identity map on $\mathcal{H}^R$ are denoted by $\openone^R$ and $\id^R$, respectively, which may be omitted if obvious.
In particular, we may write the identity operator on a $k$-dimensional Hilbert space as $\openone_k$ for clarity.
The computational basis of any $k$-dimensional Hilbert space is written as $\left\{\Ket{0},\ldots,\Ket{k-1}\right\}$.
The generalized Pauli operators on a $k$-dimensional Hilbert space are denoted by
$X_k\coloneqq\sum_{l=0}^{k-1}\Ket{l+1\bmod k}\Bra{l}$ and
$Z_k\coloneqq\sum_{l=0}^{k-1}\exp\left(\frac{\textup{i}2\pi l}{k}\right)\Ket{l}\Bra{l}$.
A maximally entangled state with Schmidt rank $K$ is denoted by
$\Ket{\Phi_K}\coloneqq\frac{1}{\sqrt{K}}\sum_{l=0}^{K-1}\Ket{l}\otimes\Ket{l}$.

We repeat Eqs.~\eqref{eq:s} and~\eqref{eq:phi} in the main text for readability
\begin{equation}
  \begin{aligned}
    \Ket{\psi_0}^{AB}\coloneqq&\sqrt{\frac{2}{11}}\Ket{\Phi_2}^{AB}
    \oplus\sqrt{\frac{9}{11}}\Ket{\Phi_9}^{AB},\\
    \Ket{\psi_1}^{AB}\coloneqq&\sqrt{\frac{2}{11}}\gamma_1 X_2^A\Ket{\Phi_2}^{AB}
    \oplus\sqrt{\frac{9}{11}}{\left(X_9^A\right)}^3\Ket{\Phi_9}^{AB},\\
    \Ket{\psi_2}^{AB}\coloneqq&\sqrt{\frac{2}{11}}\gamma_2 Z_2^A\Ket{\Phi_2}^{AB}
    \oplus\sqrt{\frac{9}{11}}{\left(X_9^A\right)}^6\Ket{\Phi_9}^{AB},
  \end{aligned}
  \tag{\ref{eq:s}}
\end{equation}
\begin{equation}
  \Ket{\psi}\coloneqq\frac{1}{\sqrt{3}}\sum_{l=0}^{2}\Ket{l}^R\otimes\Ket{\psi_l}^{AB}.
  \tag{\ref{eq:phi}}
\end{equation}
For clarity, we note that vectors on a direct sum of Hilbert spaces such as the right hand sides of Eq.~\eqref{eq:s} in the main text
\begin{equation}
  \left(\sum_{m=0}^{1}\alpha_{m}\Ket{m}\right)\oplus\left(\sum_{m=0}^{8}\beta_{m}\Ket{m}\right)\in\mathbb{C}^2\oplus\mathbb{C}^9
\end{equation}
can be regarded as
\begin{equation}
  \left(\sum_{m=0}^{1}\alpha_{m}\Ket{m}\right)+\left(\sum_{m=2}^{10}\beta_{m-2}\Ket{m}\right)\in\mathbb{C}^{11}.
\end{equation}

\section{One-way LOCC protocol achieving non-catalytic exact state merging at one ebit of entanglement cost}

While we show in the main text that non-catalytic approximate state merging of $\Ket{\psi}^{RAB}$ defined as Eq.~\eqref{eq:phi} in the main text cannot be achieved at zero entanglement cost and within a sufficiently small but nonzero error, we construct here a one-way LOCC protocol for this task achieving one ebit of entanglement cost and zero error, which yields the first statement of Theorem~\ref{thm:result} in the main text.
Note that this one-way LOCC protocol is less costly than a trivial protocol performing quantum teleportation of $A$'s part of $\Ket{\psi}^{RAB}$ of an eleven-dimensional system.
Our one-way LOCC protocol is based on the protocol for non-catalytic exact state merging established in Ref.~\cite{Y12}, which uses a decomposition theorem called the Koashi-Imoto decomposition~\cite{K3,H6,K5,W4}.
While the general protocol shown in Ref.~\cite{Y12} requires $\log_2 3$ ebits of entanglement cost for $\Ket{\psi}^{RAB}$, we modify this protocol, using a specific structure of $\Ket{\psi}^{RAB}$, to achieve one ebit of entanglement cost.
In the following, we mainly discuss this specific part in our case of non-catalytic exact state merging of $\Ket{\psi}^{RAB}$, and regarding construction of the general protocol, refer to Ref.~\cite{Y12}.

The Koashi-Imoto decomposition is shown as follows.
Note that the decomposition corresponding to that called maximal in Ref~\cite{K3} is \textit{uniquely} determined.
Refer to Ref.~\cite{Y12} for how to obtain this decomposition.

\begin{lemma}
\label{lem:koashi_imoto_decomposition_tripartite}
  (Lemma~11 in Ref.~\cite{W4}, in Section~II~C in Ref.~\cite{Y12})
  \textit{Koashi-Imoto decomposition of a tripartite pure state.}
  Given any pure state $\Ket{\psi}^{RAB}$,
  there exists an algorithmic procedure to obtain a unique decomposition of $\mathcal{H}^A$ and $\mathcal{H}^B$
  \begin{equation}
    \mathcal{H}^A=\bigoplus_{j=0}^{J-1}\mathcal{H}^{a_j^\textup{L}}\otimes\mathcal{H}^{a_j^\textup{R}},\quad
    \mathcal{H}^B=\bigoplus_{j=0}^{J-1}\mathcal{H}^{b_j^\textup{L}}\otimes\mathcal{H}^{b_j^\textup{R}},
  \end{equation}
  such that
  $\Ket{\psi}^{RAB}$ is uniquely decomposed into
  \begin{equation}
    \Ket{\psi}^{RAB}=\bigoplus_{j=0}^{J-1}\sqrt{p\left(j\right)}\Ket{\omega_j}^{a_j^\textup{L} b_j^\textup{L}}\otimes\Ket{\phi_j}^{R a_j^\textup{R} b_j^\textup{R}},
  \end{equation}
  where $p\left(j\right)$ is a probability distribution.
\end{lemma}

Using the Koashi-Imoto decomposition, we obtain the following.

\begin{proposition}
  \textit{Entanglement cost in a non-catalytic exact state merging by one-way LOCC.}
  There exists a one-way LOCC protocol for non-catalytic exact state merging of $\Ket{\psi}^{RAB}$ defined as Eq.~\eqref{eq:phi} in the main text achieving
  \begin{equation}
    \label{eq:one_ebit}
    \log_2 K=1,
  \end{equation}
  where non-catalytic exact state merging corresponds to zero error $\epsilon=0$.
\end{proposition}

\begin{proof}
The proof is by construction, and we present a one-way LOCC protocol achieving Eq.~\eqref{eq:one_ebit}.
For brevity, we define
\begin{align}
    &\Ket{\Psi_0}\coloneqq\Ket{\Phi_2},\\
    &\Ket{\Psi_1}\coloneqq\left(\gamma_1 X_2^A\otimes\openone^B\right)\Ket{\Phi_2},\\
    &\Ket{\Psi_2}\coloneqq\left(\gamma_2 Z_2^A\otimes\openone^B\right)\Ket{\Phi_2}.
\end{align}

Using Lemma~\ref{lem:koashi_imoto_decomposition_tripartite},
we obtain the following Koashi-Imoto decomposition of $\Ket{\psi}^{RAB}$.
The Hilbert spaces $\mathcal{H}^A=\mathbb{C}^{11}$ of $A$ and $\mathcal{H}^B=\mathbb{C}^{11}$ of $B$ are decomposed into
\begin{equation}
  \label{eq:h_decomposition}
  \mathcal{H}^A=\bigoplus_{j=0}^{3}\mathcal{H}^{a_j^\textup{R}},\quad
  \mathcal{H}^B=\bigoplus_{j=0}^{3}\mathcal{H}^{b_j^\textup{R}},
\end{equation}
where
\begin{align}
    &\dim\mathcal{H}^{a_0^\textup{R}}=\dim\mathcal{H}^{b_0^\textup{R}}=2,\\
    &\dim\mathcal{H}^{a_1^\textup{R}}=\dim\mathcal{H}^{b_1^\textup{R}}=3,\\
    &\dim\mathcal{H}^{a_2^\textup{R}}=\dim\mathcal{H}^{b_2^\textup{R}}=3,\\
    &\dim\mathcal{H}^{a_3^\textup{R}}=\dim\mathcal{H}^{b_3^\textup{R}}=3.
\end{align}
Note that for each $j\in\left\{0,\ldots,3\right\}$, $\mathcal{H}^{a_j^\textup{L}}$ and $\mathcal{H}^{b_j^\textup{L}}$ in Lemma~\ref{lem:koashi_imoto_decomposition_tripartite} do not explicitly appear in the decomposition in Eq.~\eqref{eq:h_decomposition}, since $\mathcal{H}^{a_j^\textup{L}}=\mathbb{C}$ and $\mathcal{H}^{b_j^\textup{L}}=\mathbb{C}$ in this case.
The state $\Ket{\psi}^{RAB}$ is decomposed into
\begin{equation}
    \Ket{\psi}^{RAB}=\sqrt{\frac{2}{11}}\Ket{\phi_0}^{Ra_0^\textup{R} b_0^\textup{R}}\oplus\bigoplus_{j=1}^{3}\sqrt{\frac{3}{11}}\Ket{\phi_j}^{Ra_j^\textup{R} b_j^\textup{R}},
\end{equation}
where
\begin{equation}
    \Ket{\phi_0}^{Ra_0^\textup{R} b_0^\textup{R}}\coloneqq\sqrt{\frac{1}{3}}\sum_{l=0}^{2}\Ket{l}^R\otimes\Ket{\Psi_l}^{a_0^\textup{R} b_0^\textup{R}},
\end{equation}
and for each $j\in\left\{1,2,3\right\}$,
\begin{equation}
  \Ket{\phi_j}^{Ra_j^\textup{R} b_j^\textup{R}}\coloneqq\sqrt{\frac{1}{9}}\sum_{l,m=0}^{2}\Ket{l}^R\otimes\Ket{l+m\bmod 3}^{a_j^\textup{R}}\otimes\Ket{m}^{b_j^\textup{R}}.
\end{equation}
While our definition of ${\left\{\Ket{\psi_l}^{AB}\right\}}_{l=0,1,2}$ in Eq.~\eqref{eq:s} in the main text uses the decomposition of each system $\mathbb{C}^{11}=\mathbb{C}^2\oplus\mathbb{C}^9$,
$\mathcal{H}^{a_0^\textup{R}}$ and $\mathcal{H}^{b_0^\textup{R}}$ in Eq.~\eqref{eq:h_decomposition} correspond to $\mathbb{C}^2$, $\mathcal{H}^{a_1^\textup{R}}$ and $\mathcal{H}^{b_1^\textup{R}}$ in Eq.~\eqref{eq:h_decomposition} correspond to a three-dimensional subspace of $\mathbb{C}^9$ spanned by $\left\{\Ket{0},\Ket{3},\Ket{6}\right\}$,
$\mathcal{H}^{a_2^\textup{R}}$ and $\mathcal{H}^{b_2^\textup{R}}$ correspond to that by $\left\{\Ket{1},\Ket{4},\Ket{7}\right\}$,
and $\mathcal{H}^{a_3^\textup{R}}$ and $\mathcal{H}^{b_3^\textup{R}}$ correspond to that by $\left\{\Ket{2},\Ket{5},\Ket{8}\right\}$.
Introducing auxiliary systems $\mathcal{H}^{a_0}$ of $A$ and $\mathcal{H}^{b_0}$ of $B$,
we can also write this decomposition as
\begin{equation}
  \label{eq:psi_decomposition}
  \begin{split}
    &\left(U^A\otimes U^B\right)\Ket{\psi}^{RAB}\\
    &=\sqrt{\frac{2}{11}}\Ket{0}^{a_0}\otimes\Ket{0}^{b_0}\otimes\Ket{\phi_0}^{Ra^\textup{R} b^\textup{R}}\\
    &\quad +\sum_{j=1}^{3}\sqrt{\frac{3}{11}}\Ket{j}^{a_0}\otimes\Ket{j}^{b_0}\otimes\Ket{\phi_j}^{Ra^\textup{R} b^\textup{R}}
  \end{split}
\end{equation}
where
\begin{align}
    &\dim\mathcal{H}^{a_0}=\dim\mathcal{H}^{b_0}=4,\\
    &\dim\mathcal{H}^{a^\textup{R}}=\max_j\left\{\dim\mathcal{H}^{a_j^\textup{R}}\right\}=3,\\
    &\dim\mathcal{H}^{b^\textup{R}}=\max_j\left\{\dim\mathcal{H}^{b_j^\textup{R}}\right\}=3,
\end{align}
$U^A$ is $A$'s local isometry from $\mathcal{H}^A$ to $\mathcal{H}^{a_0}\otimes\mathcal{H}^{a^\textup{R}}$, and
$U^B$ is $B$'s local isometry from $\mathcal{H}^B$ to $\mathcal{H}^{b_0}\otimes\mathcal{H}^{b^\textup{R}}$.

Using the Koashi-Imoto decomposition in the form of Eq.~\eqref{eq:psi_decomposition}, Ref.~\cite{Y12} shows a one-way LOCC protocol for non-catalytic exact state merging.
In this protocol, three subprocesses $1$, $2$, and $3$ shown in Ref.~\cite{Y12} are combined using controlled measurements and controlled isometries, which are controlled by states of $\mathcal{H}^{a_0}$ and $\mathcal{H}^{b_0}$.
In the following, we discuss these three subprocesses in our case.
In particular, we modify Subprocess~2 using a specific structure of $\Ket{\psi}^{RAB}$ to achieve one ebit of entanglement cost.

\textit{Subprocess 1:} The first subprocess is concerned with reduced states on $\mathcal{H}^{a_j^\textup{L}}\otimes\mathcal{H}^{b_j^\textup{L}}$, and since $\mathcal{H}^{a_j^\textup{L}}$ and $\mathcal{H}^{b_j^\textup{L}}$ do not explicitly appear in the decomposition in Eq.~\eqref{eq:h_decomposition}, we do not perform this subprocess.

\textit{Subprocess 2:} The second subprocess is for transferring $A$'s part of $\Ket{\phi_j}^{Ra^\textup{R} b^\textup{R}}$ to $B$, so that $\Ket{\phi_j}^{R{\left(b^\prime\right)}^\textup{R} b^\textup{R}}$ can be obtained, where $\mathcal{H}^{{\left(b^\prime\right)}^\textup{R}}$ is $B$'s auxiliary system corresponding to $\mathcal{H}^{a^\textup{R}}$.
While Ref.~\cite{Y12} uses quantum teleportation in this subprocess to provide a general protocol,
there may exist cases where this subprocess can be achieved at less entanglement cost than performing quantum teleportation, as pointed out in Implication~4 in Ref.~\cite{Y12}.
As for our case, $\Ket{\phi_0}^{Ra^\textup{R} b^\textup{R}}$ is merged using quantum teleportation, which requires one ebit of an initially shared maximally entangled state $\Ket{\Phi_2}^{\overline{A}\overline{B}}$, where $\mathcal{H}^{\overline{A}}$ and $\mathcal{H}^{\overline{B}}$ are systems for the shared maximally entangled states of $A$ and $B$, respectively.
If $\Ket{\phi_1}^{Ra^\textup{R} b^\textup{R}}$, $\Ket{\phi_2}^{Ra^\textup{R} b^\textup{R}}$, or $\Ket{\phi_3}^{Ra^\textup{R} b^\textup{R}}$ are also merged in the same way,
$\log_2 3$ ebits are required.
In contrast, by performing $A$'s measurement on $\mathcal{H}^{a^\textup{R}}$ in the computational basis ${\left\{\Ket{m}^{a^\textup{R}}\right\}}_{m=0,1,2}$ followed by $B$'s isometry correction conditioned by $A$'s measurement outcome, no entanglement is required for merging $\Ket{\phi_1}^{Ra^\textup{R} b^\textup{R}}$, $\Ket{\phi_2}^{Ra^\textup{R} b^\textup{R}}$, and $\Ket{\phi_3}^{Ra^\textup{R} b^\textup{R}}$.
However, to coherently combine Subprocess~2 for $\Ket{\phi_0}^{Ra^\textup{R} b^\textup{R}}$, $\Ket{\phi_1}^{Ra^\textup{R} b^\textup{R}}$, $\Ket{\phi_2}^{Ra^\textup{R} b^\textup{R}}$, and $\Ket{\phi_3}^{Ra^\textup{R} b^\textup{R}}$,
one ebit of entanglement $\Ket{\Phi_2}^{\overline{A}\overline{B}}$ has to be consumed by $A$'s measurement on $\mathcal{H}^{\overline{A}}$ in the computational basis ${\left\{\Ket{m}^{\overline{A}}\right\}}_{m=0,1}$ followed by $B$'s isometry correction.
Consequently,
the LOCC map for Subprocess~2 can be written as a family of operators ${\left\{\Bra{j,m_2}\otimes\sigma_{j,m_2}\right\}}_{m_2}$ tracing out the post-measurement state of $A$, where $\Ket{0,m_2}$ and $\sigma_{0,m_2}$ corresponds to ${\left(U_j^\prime\right)}^\dag\Ket{\Phi_{j,m_2}}$ and $\sigma_{j,m_2}$ in Subprocess~2 shown in Ref.~\cite{Y12} based on quantum teleportation, and for each $j\in\left\{1,2,3\right\}$, ${\left\{\Ket{j,m_2}\right\}}_{m_2}$ and $\sigma_{j,m_2}$ are the computational basis for $A$'s measurement and the isometry for $B$'s correction conditioned by $A$'s measurement outcome $m_2$, respectively.

\textit{Subprocess 3:} The third subprocess is for merging states on $\mathcal{H}^{a_0}\otimes\mathcal{H}^{b_0}$, and this subprocess can be performed in the same way as Ref.~\cite{Y12}.

Combining these three subprocesses in the same way as Ref.~\cite{Y12}, we obtain the one-way LOCC protocol achieving Eq.~\eqref{eq:one_ebit}.

\end{proof}

\section{Two-way LOCC protocol achieving non-catalytic exact state merging at zero entanglement cost}

We provide a detailed description of the two-way LOCC protocol achieving non-catalytic exact state merging of $\Ket{\psi}$ defined as Eq.~\eqref{eq:phi} in the main text at zero entanglement cost, which yields the second statement of Theorem~\ref{thm:result} in the main text.
As discussed in the main text, protocols achieving local state discrimination using elimination, such as the two-way LOCC protocols for $2$-LOCC sets in Refs.~\cite{N4,T7,T8}, are not applicable to state merging in a straightforward way.
For example, consider a set of mutually orthogonal states $\big\{\Ket{\psi_0}^{AB}\coloneqq\Ket{0}^A\otimes\Ket{0}^B,\Ket{\psi_1}^{AB}\coloneqq\Ket{0}^A\otimes\Ket{1}^B,\Ket{\psi_2}^{AB}\coloneqq\Ket{1}^A\otimes\Ket{+}^B\big\}$, where $\Ket{+}\coloneqq\frac{1}{\sqrt{2}}\left(\Ket{0}+\Ket{1}\right)$.
If $A$ eliminates some of the possibilities by a measurement in basis $\left\{\Ket{0},\Ket{1}\right\}$, $B$'s local measurement conditioned by $A$'s outcome can discriminate the remaining mutually orthogonal states of $B$.
However, state merging of the corresponding tripartite state $\frac{1}{\sqrt{3}}\sum_{l=0}^{2}\Ket{l}^R\otimes\Ket{\psi_l}^{AB}$ is not achievable at zero entanglement cost due to the converse bound shown in Ref.~\cite{Y12}.
In contrast to such protocols for local state discrimination using elimination, our two-way LOCC protocol for state merging can also achieve local state discrimination of ${\left\{\Ket{\psi_l}\right\}}_{l=0,1,2}$ defined as Eq.~\eqref{eq:s} in the main text \textit{without elimination}.
In particular, we show the following.

\begin{proposition}
  \textit{Entanglement cost in a non-catalytic exact state merging by two-way LOCC.}
  There exists a two-way LOCC protocol for non-catalytic exact state merging of $\Ket{\psi}^{RAB}$ defined as Eq.~\eqref{eq:phi} in the main text achieving
  \begin{equation}
    \label{eq:zero}
    \log_2 K=0,
  \end{equation}
  where non-catalytic exact state merging corresponds to zero error $\epsilon=0$.
\end{proposition}

\begin{proof}
  The proof is by construction, and we present a two-way LOCC protocol for non-catalytic exact state merging of $\Ket{\psi}$ achieving Eq.~\eqref{eq:zero}, using the notations introduced in the proof of the second statement in Theorem~\ref{thm:result} in the main text.
  In particular, we show $B$'s measurement ${\left\{M_j^B\right\}}_{j=0,1,2}$ and $A$'s measurement ${\left\{M_{k|j}^A\right\}}_{k=0,\ldots,32}$ conditioned by $B$'s measurement outcome $j$,
  where these measurement (Kraus) operators satisfy
  \begin{align}
    &\sum_{j=0}^2 {M_j^B}^\dag M_j^B=\openone,\\
    &\sum_{k=0}^{32} {M_{k|j}^A}^\dag M_{k|j}^A=\openone, \quad\forall j\in\left\{0,1,2\right\}.
  \end{align}
  In the following, $\mathcal{H}^A$ and $\mathcal{H}^B$ are decomposed in the same way as Eq.~\eqref{eq:s} in the main text for defining ${\left\{\Ket{\psi_l}^{AB}\right\}}_l$, that is,
  \begin{equation}
    \mathcal{H}^A=\mathbb{C}^2\oplus\mathbb{C}^9,\quad
    \mathcal{H}^B=\mathbb{C}^2\oplus\mathbb{C}^9.
  \end{equation}

  The measurement ${\left\{M_j^B\right\}}_{j=0,1,2}$ performed by $B$ is
  \begin{align}
    &M_0^B\coloneqq\sqrt{\frac{1}{3}}\left(\Ket{0}\Bra{0}+\Ket{1}\Bra{1}\right)\oplus\left(\Ket{0}\Bra{0}+\Ket{1}\Bra{1}+\Ket{2}\Bra{2}\right),\\
    &M_1^B\coloneqq\sqrt{\frac{1}{3}}\left(\Ket{0}\Bra{0}+\Ket{1}\Bra{1}\right)\oplus\left(\Ket{3}\Bra{3}+\Ket{4}\Bra{4}+\Ket{5}\Bra{5}\right),\\
    &M_2^B\coloneqq\sqrt{\frac{1}{3}}\left(\Ket{0}\Bra{0}+\Ket{1}\Bra{1}\right)\oplus\left(\Ket{6}\Bra{6}+\Ket{7}\Bra{7}+\Ket{8}\Bra{8}\right),
  \end{align}
  where each operator on the right-hand side is on $\mathbb{C}^2\oplus\mathbb{C}^9$.
  This measurement satisfies the completeness condition
  \begin{equation}
    \sum_{j=0}^{2}M_j^\dag M_j=\openone.
  \end{equation}

  We show $A$'s measurement ${\left\{M_{k|j}^A\right\}}_{k=0,\ldots,32}$ conditioned by $j\in\left\{0,1,2\right\}$.
  We first show the case of $j=0$, that is, ${\left\{M_{k|0}^A\right\}}_{k=0,\ldots,32}$, while a similar  construction applies to the cases of $j=1,2$ as discussed later.
  For brevity, we define a bipartite pure state $\Ket{\Psi}\in\mathbb{C}^9\otimes\mathbb{C}^9$ with Schmidt rank three as
  \begin{equation}
    \Ket{\Psi}\coloneqq\sqrt{\frac{1}{3}}\left(\Ket{0}\otimes\Ket{0}+\Ket{1}\otimes\Ket{1}+\Ket{2}\otimes\Ket{2}\right),
  \end{equation}
  and we also define the Fourier-basis states of three-dimensional subspaces of $\mathbb{C}^9$
  \begin{align}
    \Ket{\omega_{n}^{\left(0,4,8\right)}}&\coloneqq\frac{1}{\sqrt{3}}\Ket{0}+\frac{\exp\left(\frac{\textup{i}\pi n}{3}\right)}{\sqrt{3}}\Ket{4}+\frac{\exp\left(\frac{\textup{i}\pi 2n}{3}\right)}{\sqrt{3}}\Ket{8},\\
    \Ket{\omega_{n}^{\left(1,5,6\right)}}&\coloneqq\frac{1}{\sqrt{3}}\Ket{1}+\frac{\exp\left(\frac{\textup{i}\pi n}{3}\right)}{\sqrt{3}}\Ket{5}+\frac{\exp\left(\frac{\textup{i}\pi 2n}{3}\right)}{\sqrt{3}}\Ket{6},\\
    \Ket{\omega_{n}^{\left(2,3,7\right)}}&\coloneqq\frac{1}{\sqrt{3}}\Ket{2}+\frac{\exp\left(\frac{\textup{i}\pi n}{3}\right)}{\sqrt{3}}\Ket{3}+\frac{\exp\left(\frac{\textup{i}\pi 2n}{3}\right)}{\sqrt{3}}\Ket{7},
  \end{align}
  where $n\in\left\{0,1,2\right\}$.
  If $B$'s measurement outcome is $j=0$, the post-measurement state is
  \begin{equation}
    \Ket{\psi^{\left(0\right)}}^{RAB}=\frac{1}{\sqrt{3}}\sum_{l=0}^{2}\Ket{l}^R\otimes\Ket{\psi_l^{\left(0\right)}}^{AB},
  \end{equation}
  where
  \begin{align}
    \begin{split}
      \Ket{\psi_0^{\left(0\right)}}\coloneqq&\sqrt{\frac{2}{11}}\Ket{\Phi_2}\oplus\sqrt{\frac{9}{11}}\Ket{\Psi}\\
      =&\sqrt{\frac{1}{11}}\left(\Ket{0}\otimes\Ket{0}+\Ket{1}\otimes\Ket{1}\right)\oplus\\
       &\sqrt{\frac{3}{11}}\left(\Ket{0}\otimes\Ket{0}+\Ket{1}\otimes\Ket{1}+\Ket{2}\otimes\Ket{2}\right),
    \end{split}\\
    \begin{split}
      \Ket{\psi_1^{\left(0\right)}}\coloneqq&\sqrt{\frac{2}{11}}\left(\gamma_1 X_2\otimes\openone\right)\Ket{\Phi_2}\oplus\sqrt{\frac{9}{11}}\left({\left(X_9\right)}^3\otimes\openone\right)\Ket{\Psi}\\
      =&\sqrt{\frac{1}{11}}\gamma_1\left(\Ket{1}\otimes\Ket{0}+\Ket{0}\otimes\Ket{1}\right)\oplus\\
       &\sqrt{\frac{3}{11}}\left(\Ket{3}\otimes\Ket{0}+\Ket{4}\otimes\Ket{1}+\Ket{5}\otimes\Ket{2}\right),
    \end{split}\\
    \begin{split}
      \Ket{\psi_2^{\left(0\right)}}\coloneqq&\sqrt{\frac{2}{11}}\left(\gamma_2 Z_2\otimes\openone\right)\Ket{\Phi_2}\oplus\sqrt{\frac{9}{11}}\left({\left(X_9\right)}^6\otimes\openone\right)\Ket{\Psi}\\
      =&\sqrt{\frac{1}{11}}\gamma_2\left(\Ket{0}\otimes\Ket{0}-\Ket{1}\otimes\Ket{1}\right)\oplus\\
       &\sqrt{\frac{3}{11}}\left(\Ket{6}\otimes\Ket{0}+\Ket{7}\otimes\Ket{1}+\Ket{8}\otimes\Ket{2}\right).
    \end{split}
  \end{align}
  Note that in these definitions of $\Ket{\psi_0^{\left(0\right)}}$, $\Ket{\psi_1^{\left(0\right)}}$, and $\Ket{\psi_2^{\left(0\right)}}$, $A$'s parts of the second terms are on mutually orthogonal three-dimensional subspaces of $\mathbb{C}^9$ spanned by $\left\{\Ket{0},\Ket{1},\Ket{2}\right\}$, $\left\{\Ket{3},\Ket{4},\Ket{5}\right\}$, and $\left\{\Ket{6},\Ket{7},\Ket{8}\right\}$, respectively.
  In this case, $A$'s measurement ${\left\{M_{k|0}^A\right\}}_{k=0,\ldots,32}$ is in the form of
  \begin{equation}
    M_{k|0}\coloneqq\Bra{\phi_{k|0}},
  \end{equation}
  where $k\in\left\{0,\ldots,32\right\}$, the post-measurement state of $A$ is traced out, and $\Ket{\phi_{k|0}}\in\mathbb{C}^2\oplus\mathbb{C}^9$ is an unnormalized vector.
  Each $\Ket{\phi_{k|0}}$ is defined as
  \begin{alignat}{2}
    \Ket{\phi_{0 |0}}   &\coloneqq& \sqrt{\frac{3}{36}}\Ket{0}&\oplus\sqrt{\frac{1}{36}}\left( \Ket{0}+\Ket{4}-\overline{\gamma_2}\Ket{6}\right),\\
    \Ket{\phi_{1 |0}}   &\coloneqq&-\sqrt{\frac{3}{36}}\Ket{0}&\oplus\sqrt{\frac{1}{36}}\left( \Ket{0}+\Ket{4}-\overline{\gamma_2}\Ket{6}\right),\\
    \Ket{\phi_{2 |0}}   &\coloneqq& \sqrt{\frac{3}{36}}\Ket{1}&\oplus\sqrt{\frac{1}{36}}\left(-\Ket{0}+\Ket{4}-\overline{\gamma_2}\Ket{6}\right),\\
    \Ket{\phi_{3 |0}}   &\coloneqq&-\sqrt{\frac{3}{36}}\Ket{1}&\oplus\sqrt{\frac{1}{36}}\left(-\Ket{0}+\Ket{4}-\overline{\gamma_2}\Ket{6}\right),\\
    \Ket{\phi_{4 |0}}   &\coloneqq& \sqrt{\frac{3}{36}}\Ket{0}&\oplus\sqrt{\frac{1}{36}}\left( \Ket{0}-\Ket{4}-\overline{\gamma_2}\Ket{6}\right),\\
    \Ket{\phi_{5 |0}}   &\coloneqq&-\sqrt{\frac{3}{36}}\Ket{0}&\oplus\sqrt{\frac{1}{36}}\left( \Ket{0}-\Ket{4}-\overline{\gamma_2}\Ket{6}\right),\\
    \Ket{\phi_{6 |0}}   &\coloneqq& \sqrt{\frac{3}{36}}\Ket{1}&\oplus\sqrt{\frac{1}{36}}\left(-\Ket{0}-\Ket{4}-\overline{\gamma_2}\Ket{6}\right),\\
    \Ket{\phi_{7 |0}}   &\coloneqq&-\sqrt{\frac{3}{36}}\Ket{1}&\oplus\sqrt{\frac{1}{36}}\left(-\Ket{0}-\Ket{4}-\overline{\gamma_2}\Ket{6}\right),\\
    \Ket{\phi_{8 |0}}   &\coloneqq& \sqrt{\frac{3}{36}}\Ket{0}&\oplus\sqrt{\frac{1}{36}}\left( \Ket{1}+\Ket{5}-\overline{\gamma_2}\Ket{7}\right),\\
    \Ket{\phi_{9 |0}}   &\coloneqq&-\sqrt{\frac{3}{36}}\Ket{0}&\oplus\sqrt{\frac{1}{36}}\left( \Ket{1}+\Ket{5}-\overline{\gamma_2}\Ket{7}\right),\\
    \Ket{\phi_{10|0}}   &\coloneqq& \sqrt{\frac{3}{36}}\Ket{1}&\oplus\sqrt{\frac{1}{36}}\left(-\Ket{1}+\Ket{5}-\overline{\gamma_2}\Ket{7}\right),\\
    \Ket{\phi_{11|0}}   &\coloneqq&-\sqrt{\frac{3}{36}}\Ket{1}&\oplus\sqrt{\frac{1}{36}}\left(-\Ket{1}+\Ket{5}-\overline{\gamma_2}\Ket{7}\right),\\
    \Ket{\phi_{12|0}}   &\coloneqq& \sqrt{\frac{3}{36}}\Ket{0}&\oplus\sqrt{\frac{1}{36}}\left( \Ket{1}-\Ket{5}-\overline{\gamma_2}\Ket{7}\right),\\
    \Ket{\phi_{13|0}}   &\coloneqq&-\sqrt{\frac{3}{36}}\Ket{0}&\oplus\sqrt{\frac{1}{36}}\left( \Ket{1}-\Ket{5}-\overline{\gamma_2}\Ket{7}\right),\\
    \Ket{\phi_{14|0}}   &\coloneqq& \sqrt{\frac{3}{36}}\Ket{1}&\oplus\sqrt{\frac{1}{36}}\left(-\Ket{1}-\Ket{5}-\overline{\gamma_2}\Ket{7}\right),\\
    \Ket{\phi_{15|0}}   &\coloneqq&-\sqrt{\frac{3}{36}}\Ket{1}&\oplus\sqrt{\frac{1}{36}}\left(-\Ket{1}-\Ket{5}-\overline{\gamma_2}\Ket{7}\right),\\
    \Ket{\phi_{16|0}}   &\coloneqq& \sqrt{\frac{3}{36}}\Ket{0}&\oplus\sqrt{\frac{1}{36}}\left( \Ket{2}+\Ket{3}-\overline{\gamma_2}\Ket{8}\right),\\
    \Ket{\phi_{17|0}}   &\coloneqq&-\sqrt{\frac{3}{36}}\Ket{0}&\oplus\sqrt{\frac{1}{36}}\left( \Ket{2}+\Ket{3}-\overline{\gamma_2}\Ket{8}\right),\\
    \Ket{\phi_{18|0}}   &\coloneqq& \sqrt{\frac{3}{36}}\Ket{1}&\oplus\sqrt{\frac{1}{36}}\left(-\Ket{2}+\Ket{3}-\overline{\gamma_2}\Ket{8}\right),\\
    \Ket{\phi_{19|0}}   &\coloneqq&-\sqrt{\frac{3}{36}}\Ket{1}&\oplus\sqrt{\frac{1}{36}}\left(-\Ket{2}+\Ket{3}-\overline{\gamma_2}\Ket{8}\right),\\
    \Ket{\phi_{20|0}}   &\coloneqq& \sqrt{\frac{3}{36}}\Ket{0}&\oplus\sqrt{\frac{1}{36}}\left( \Ket{2}-\Ket{3}-\overline{\gamma_2}\Ket{8}\right),\\
    \Ket{\phi_{21|0}}   &\coloneqq&-\sqrt{\frac{3}{36}}\Ket{0}&\oplus\sqrt{\frac{1}{36}}\left( \Ket{2}-\Ket{3}-\overline{\gamma_2}\Ket{8}\right),\\
    \Ket{\phi_{22|0}}   &\coloneqq& \sqrt{\frac{3}{36}}\Ket{1}&\oplus\sqrt{\frac{1}{36}}\left(-\Ket{2}-\Ket{3}-\overline{\gamma_2}\Ket{8}\right),\\
    \Ket{\phi_{23|0}}   &\coloneqq&-\sqrt{\frac{3}{36}}\Ket{1}&\oplus\sqrt{\frac{1}{36}}\left(-\Ket{2}-\Ket{3}-\overline{\gamma_2}\Ket{8}\right),\\
    \Ket{\phi_{24|0}}&\coloneqq&\boldsymbol{0}&\oplus\sqrt{\frac{28}{36}}\Ket{\omega_{0}^{\left(0,4,8\right)}},\\
    \Ket{\phi_{25|0}}&\coloneqq&\boldsymbol{0}&\oplus\sqrt{\frac{28}{36}}\Ket{\omega_{1}^{\left(0,4,8\right)}},\\
    \Ket{\phi_{26|0}}&\coloneqq&\boldsymbol{0}&\oplus\sqrt{\frac{28}{36}}\Ket{\omega_{2}^{\left(0,4,8\right)}},\\
    \Ket{\phi_{27|0}}&\coloneqq&\boldsymbol{0}&\oplus\sqrt{\frac{28}{36}}\Ket{\omega_{0}^{\left(1,5,6\right)}},\\
    \Ket{\phi_{28|0}}&\coloneqq&\boldsymbol{0}&\oplus\sqrt{\frac{28}{36}}\Ket{\omega_{1}^{\left(1,5,6\right)}},\\
    \Ket{\phi_{29|0}}&\coloneqq&\boldsymbol{0}&\oplus\sqrt{\frac{28}{36}}\Ket{\omega_{2}^{\left(1,5,6\right)}},\\
    \Ket{\phi_{30|0}}&\coloneqq&\boldsymbol{0}&\oplus\sqrt{\frac{28}{36}}\Ket{\omega_{0}^{\left(2,3,7\right)}},\\
    \Ket{\phi_{31|0}}&\coloneqq&\boldsymbol{0}&\oplus\sqrt{\frac{28}{36}}\Ket{\omega_{1}^{\left(2,3,7\right)}},\\
    \Ket{\phi_{32|0}}&\coloneqq&\boldsymbol{0}&\oplus\sqrt{\frac{28}{36}}\Ket{\omega_{2}^{\left(2,3,7\right)}},
  \end{alignat}
  where $\boldsymbol{0}$ is the zero vector on $\mathbb{C}^2$.
  This measurement satisfies the completeness condition
  \begin{equation}
    \sum_{k=0}^{32}M_{k|0}^\dag M_{k|0}=\openone.
  \end{equation}

  Similarly,
  the other measurements for $A$ conditioned by $B$'s measurement outcomes $j=1$ and $j=2$, that is,
  ${\left\{M_{k|1}^A\right\}}_{k=0,\ldots,32}$ and ${\left\{M_{k|2}^A\right\}}_{k=0,\ldots,32}$, respectively,
  are defined for each $k\in\left\{0,\ldots,32\right\}$ as
  \begin{align}
    M_{k|1}&\coloneqq M_{k|0}\left(\boldsymbol{0}\oplus{\left(X_9\right)}^3\right),\\
    M_{k|2}&\coloneqq M_{k|0}\left(\boldsymbol{0}\oplus{\left(X_9\right)}^6\right).
  \end{align}
  These measurements satisfy the completeness condition
  \begin{equation}
    \sum_{k=0}^{32}M_{k|j}^\dag M_{k|j}=\openone,
  \end{equation}
  for each $j\in\left\{1,2\right\}$.

  In our two-way LOCC protocol for non-catalytic exact state merging of $\Ket{\psi}^{RAB}$ at zero entanglement cost,
  $B$ first performs the measurement ${\left\{M_j^B\right\}}_{j=0,1,2}$, and the measurement outcome $j$ is sent by classical communication from $B$ to $A$.
  Conditioned by $j$, the measurement ${\left\{M_{k|j}^A\right\}}_{k=0,\ldots,32}$ is performed by $A$, and the measurement outcome $k$ is sent by classical communication from $A$ to $B$.
  After this LOCC measurement ${\left\{M_{k|j}^A\otimes M_j^B\right\}}_{j,k}$ by $A$ and $B$,
  for any pair of measurement outcomes $j\in\left\{0,1,2\right\}$ and $k\in\left\{0,\ldots,32\right\}$,
  the post-measurement state
  \begin{equation}
    \frac{\left(M_{k|j}^A\otimes M_j^B\right)\Ket{\psi}^{RAB}}{\left\|\left(M_{k|j}^A\otimes M_j^B\right)\Ket{\psi}^{RAB}\right\|},
  \end{equation}
  is a maximally entangled state with Schmidt rank three between $R$ and $B$.
  Therefore, $B$ performs local isometry conditioned by $j$ and $k$ to transform this maximally entangled state into $\Ket{\psi}^{RB^\prime B}$.
  This protocol yields the conclusion.
\end{proof}

\end{document}